\documentclass{article}

\usepackage{arxiv}

\usepackage[utf8]{inputenc} 
\usepackage[T1]{fontenc}    
\usepackage{hyperref}       
\usepackage{url}            
\usepackage{booktabs}       
\usepackage{amsfonts}       
\usepackage{nicefrac}       
\usepackage{microtype}      
\usepackage{lipsum}
\usepackage{graphicx}
\graphicspath{ {./images/} }

\usepackage{algorithm}
\usepackage{algpseudocode}  

\usepackage{amsthm}
\newtheorem{theorem}{Theorem}[section]
\newtheorem{lemma}[theorem]{Lemma}

\usepackage{amsmath}

\title{Linked Array Tree: A Constant-Time Search Structure for Big Data}

\author{
 Songpeng Liu \\
  College of Systems and Society\\
  Australian National University\\
  \texttt{Songpeng.liu@outlook.com} \\
}

\begin{document}
\maketitle

\begin{abstract}
As data volumes continue to grow rapidly, traditional search algorithms, like the red-black tree and B+ Tree, face increasing challenges in performance, especially in big data scenarios with intensive storage access. This paper presents the Linked Array Tree (LAT), a novel data structure designed to achieve constant-time complexity for search, insertion, and deletion operations. LAT leverages a sparse, non-moving hierarchical layout that enables direct access paths without requiring rebalancing or data movement. Its low memory overhead and avoidance of pointer-heavy structures make it well-suited for large-scale and intensive workloads. While not specifically tested under parallel or concurrent conditions, the structure’s static layout and non-interfering operations suggest potential advantages in such environments.

This paper first introduces the structure and algorithms of LAT, followed by a detailed analysis of its time complexity in search, insertion, and deletion operations. Finally, it presents experimental results across both data-intensive and sparse usage scenarios to evaluate LAT’s practical performance.
\end{abstract}

\section{Introduction}
Search algorithms and relevant data structures are an essential foundation in all aspects of computer science \cite{1,2,3}. As they play a necessary role in building scalable, reliable, stable systems and software, from the 1950s, many high-performance search algorithms have been invented, which possess diverse virtues to satisfy different criteria. Generally, the goals of designing search algorithms and data structures are to minimize the time and space complexity and to maximize the implementation simplicity.

Traditional basic search algorithms work well when the dataset is relatively small. However, as the data size grows so fast and usage scenarios become much more complex in modern applications, several major problems emerged in big data.

This paper will first identify some problems in traditional basic search algorithms. Then we will review some widely distributed basic search algorithms. In the main sections, the paper will introduce a novel basic search algorithm and its optimization, and provide the analysis. In the end, the paper will present a performance comparison test based on in-memory search scenarios between the new search algorithm and the modern implementations of B+-tree and red-black tree.

\subsection{Problems when Searching in Big Data}
The first problem in big data search is the low storage access speed. In many environments, the big data volume is so massive that the memory cannot store all the information, so the information needs to be stored on and read from disks. Nevertheless, the storage access speed is extremely slow and has a high latency compared with the high-speed CPU calculations. According to the research \cite{5}, most latency in search algorithms is caused by storage access. If the data is stored on disk, the latency is caused by disk access; if the data is stored in memory, the memory access will cause the most latency. In big data scope, the storage access times will become much more in traditional search algorithms, and the performance of searching will be impacted enormously. There are some approaches to mitigate this problem, like prefetching, but the more foundational solution is to reduce the storage access times. 
	
The second problem is the high overhead. For example, in a red-black tree, every element would contain 3 points besides the data itself. When the data volume increases, the overhead will rise linearly. Therefore, we desire a search algorithm whose overhead is much less.
	
The third problem is the difficulty in supporting parallelization and concurrent environments. To speed up the search, scholars have been researching search algorithms in concurrent environments since the 1970s. The main challenge is caused by the moving elements in inserting and deleting operations. If two processes operate the same part at the same time, it may cause the cache coherence problem. The main approach is to add locks when operating the sensitive part of the structure. Unfortunately, locks will impact the performance. Generally, the higher the element moving frequency, the higher the utilization of locks, and the higher the impact on the performance.

\subsection{Related Work}
In general solutions, logarithmic time complexity in searching has been achieved by tree search algorithms: AVL tree \cite{6}, red-black tree \cite{7}, B-tree and its variations \cite{8,9}, and index search algorithm: skip list \cite{10}.

In the early era of computer science history, a series of basic data structures were invented, including the array and linked list. However, none of them can offer virtues that can tackle large-scale and complex problems. Arrays need a continuous chunk of memory, while linked lists only offer a linear time complexity. Hence, in 1960, P. F. Windley proposed a new data structure called a tree, which can reduce the search time complexity to logarithmic \cite{11}.

However, the original tree structure does not introduce any balancing mechanism. So, in the worst case, the tree could be like a linked list. In 1962, G. M. Adel’son-Vel’skii and Y. M. Landis invented the earliest self-balancing binary search tree data structure, which now is called the AVL tree. In the original algorithm, in each insertion and deletion, after the node is added or removed, all path depths then would be calculated. If the difference between the path depths is more than 1, this means the tree is not balanced, and the tree needs to be rotated.

However, even though the AVL tree can reach a logarithmic search complexity, its balance mechanism needs information from other paths to calculate the depth difference. In 1978, LJ Guibas and R Sedgewick introduced a dichromatic approach, where the colors of nodes indicate the balance properties, to address this challenge, and the data structure is the red-black tree, the most widely distributed binary search tree algorithm nowadays.

When the data set is big, even if the structure is a perfectly balanced binary search tree, it still needs to travel through many different nodes and access the memory many times. When searching on disks, this problem becomes prominent. B-tree and its variants, invented by R. Bayer and E. McCreight in 1970, solved part of the problem by adopting a block load mechanism. Because of their lower loading times, they are ubiquitously employed in modern file systems and database systems. Even though the B+-tree can reduce the disk access times, its memory access times are still high because it needs to read every node in the block to compare the key. In the worst case, it need to traverse all nodes in all blocks on the path, e.g., retrieving the latest element.

Besides these most used basic search algorithms, many other profound search algorithms and data structures have been invented \cite{14,15,16}. However, after we entered the 21st century, the general research interests shifted from designing basic search algorithms to innovatively applying them, such as using CUDA to implement B+-tree \cite{4}, or to improve the performance of concurrent B+-tree \cite{12}.

As we can see, the traditional search algorithms are based on value comparison, comparing the value and traveling to the next node, and repeating the process. Mathematics in computer science is discrete. When dealing with discrete data, the search could be based on value calculation instead of value comparison to speed up the process and reduce the storage access times.To the best of our knowledge, LAT is the first search algorithm that achieves constant-time access with low memory overhead in practical big data settings

\section{Linked Array Tree}
This section will introduce a new search algorithm, called Linked Array Tree (LAT). It has a similar structure to B+-tree, but the management process and the search concept have huge discrepancies.

\textbf{Terms:} 
\begin{itemize}
    \item \textbf{Key}: The unique identifier assigned to an item.
    \item \textbf{Value}: The data of the item, each value should be paired with a key.
    \item \textbf{Node}: The basic unit that stores a pointer or a value.
    \item \textbf{Height}: The amount of hierarchy layers in the structure.
    \item \textbf{Level $x$}: Hierarchy layer sequence, the top level is level $0$; the last level is level $height-1$.
    \item \textbf{Data Level}: The last level which is used to store value.
    \item \textbf{Index Level}: All levels other than the data level are index levels.
    \item \textbf{Data Array}: Arrays in the data level that have $radix$ nodes and a pointer, and the nodes store values.
    \item \textbf{Index Array}: Arrays in index level that have $radix$ nodes, and the nodes store pointers.
    \item \textbf{Radix}: The number of nodes in an array, should be bigger than 1.
    \item \textbf{Max size}: The maximum quantity of the item that the present structure could store.
    \item \textbf{Remainder $x$}: A number that is used to locate the node in the array in level $x$.
\end{itemize}

The LAT has a similar multi-level pyramid structure to the B+-tree. The $values$ are orderly stored at the $data level$. But in $index levels$, the index arrays only store pointers to the next level. Each data array contains an extra pointer pointing to the next data array besides the $values$. The $key$ is not been explicitly stored in the structure. When initializing an LAT, a $radix$ and $height$ should be set. Figure 1 presents an example of a LAT whose $radix$ is 4 and $height$ is 4. When the $radix$ is 2, LAT could be seen as a perfectly balanced binary search tree.

\begin{figure}[htbp]
\centerline{\includegraphics[width=0.8\linewidth]{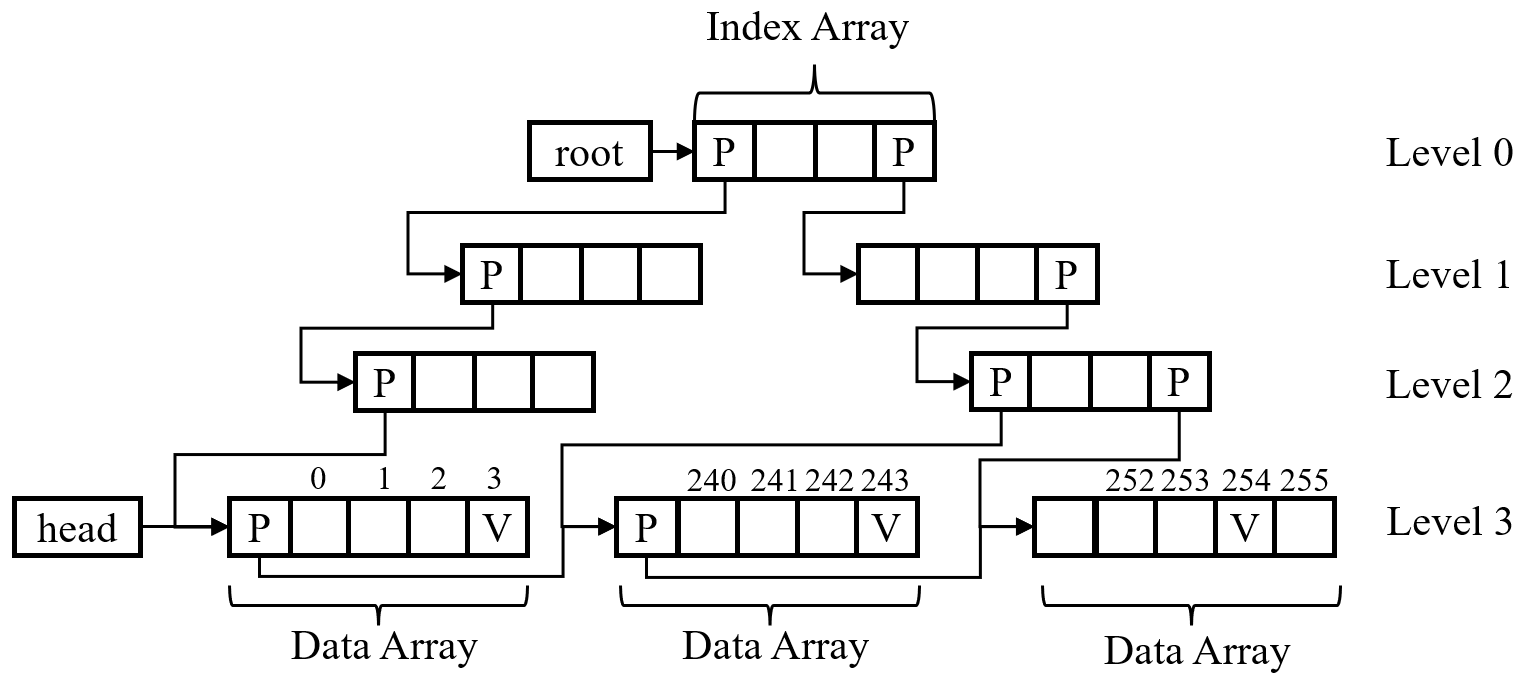}}
\caption{A Linked Array Tree whose $radix$ is 4 and $height$ is 4.}
\end{figure}

\begin{lemma}
Let $r$ be the $radix$, $h$ be the $height$, $m$ be the max size of the LAT, then $m=r^h$.
\end{lemma}

\begin{proof}
Let $l_{n}$ be the maximum number of nodes that $level\ n$ can store, then $l_{n}=l_{n-1} \times r$ because each pointer points to a new array which can store $r$ nodes. Since $l_{0}=r$, according to the mathematics induction, $l_{n}=r^{n+1}$. Since the data level is at $level\ h-1$, the max size of the structure is $m=l_{h-1}=r^h$ 
\end{proof}

There is an approach to make the max size and structure dynamic. It is simply to add another level on top, and the max size would be multiplied by $radix$. Figure 2 illustrates the process.
\begin{figure}[htbp]
\centerline{\includegraphics[width=\linewidth]{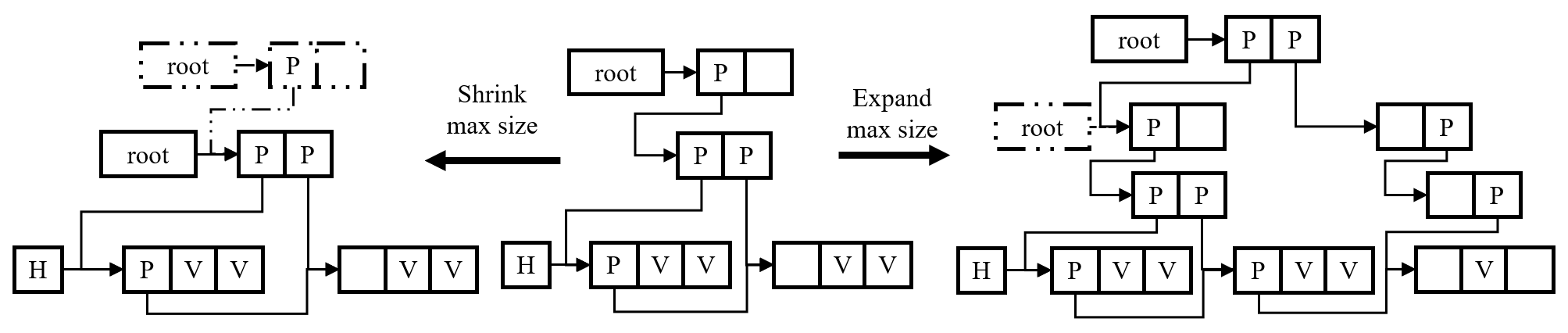}}
\caption{Changing the max size of an LAT.}
\end{figure}

\subsection{Search}
When searching for a value, after the key is inputted, a series of remainders would be calculated out based on the $radix$ and $height$. Algorithm 1 records the process. The recommended structure used to store remainders is a stack since the remainders that are calculated out earlier will be used later.

\begin{figure*}[htbp]
\centerline{\includegraphics[width=\linewidth]{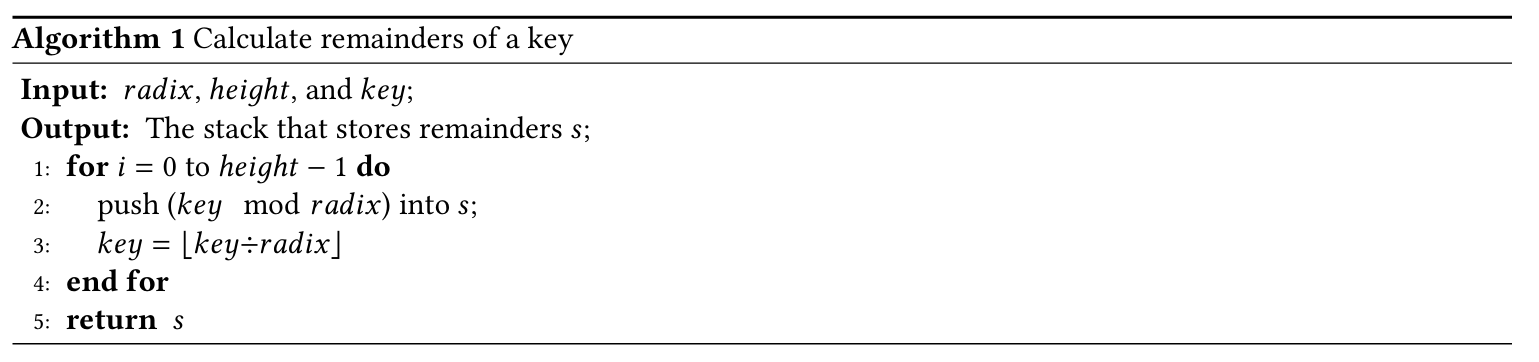}}
\end{figure*}

For example, in Figure 1, the remainders of $key\ 254$ should be: $3,3,3,2$. The remainders are used to locate the search path directly. A remainder is an index to access the pointer in the index array to travel to the next level. For example, we firstly pop 3 out from the stack as $remainder\ 0$, we directly access the fourth pointer in the $level 0$. We repeat the process until we reach the data level. An empty node means the value of the key does not exist.

\subsection{Insertion}
If the data array already exists, we can simply travel to the location via the search process and set the $value$ in it. If we encounter an empty node, a new array needs to be allocated. If the $data array$ is newly allocated, it needs to be inserted into the array list. This needs us to find the left data array of it. To ensure efficiency, we can not search it from the head. We can locate the left array with a recursive search from the root. Figure 3 shows the search path of a newly allocated data array, and Algorithm 2 records the process.
\begin{figure}[htbp]
\centerline{\includegraphics[width=0.6\linewidth]{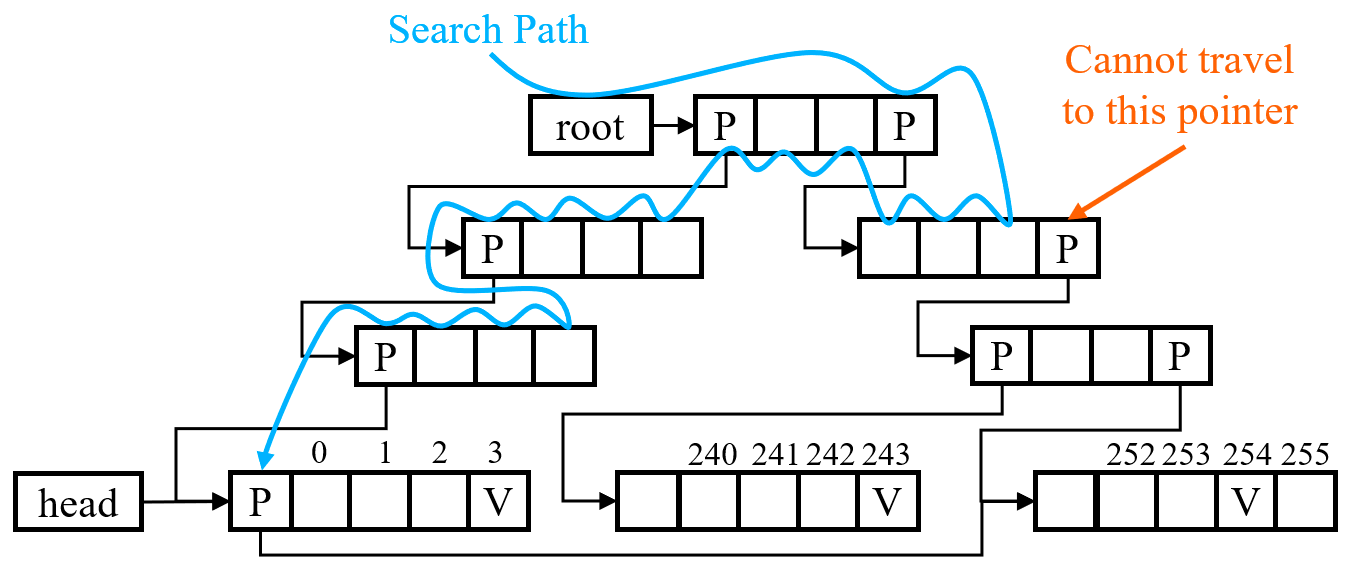}}
\caption{Finding the left data array.}
\end{figure}

\begin{figure*}[htbp]
\centerline{\includegraphics[width=\linewidth]{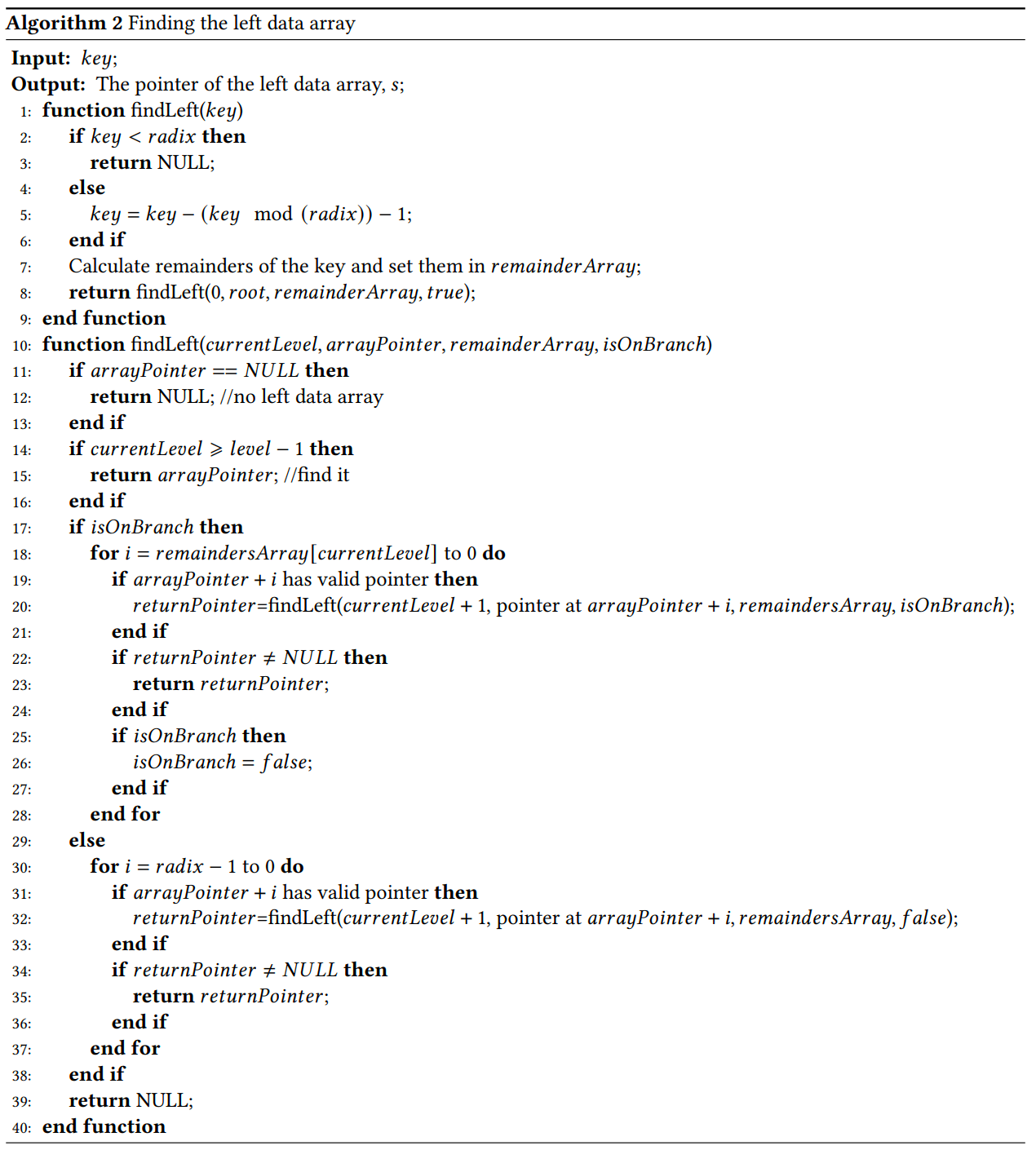}}
\end{figure*}

Several problems need to be solved when finding the left data array. The first problem is to prevent the process from accessing the newly allocated array itself. This challenge is solved by a key conversion in Algorithm 2. Through $key = key- (key\mod(radix)) -1$, we can ensure the process starts the finding from left, instead from the newly allocated array itself. 

There is still another problem: whether the present array is on the search branch? If the array is on the search branch, it may only need to iterate over a part of the pointers in it. When the process travels to a new level, it needs information from the above level to judge. In Algorithm 2, through the boolean parameter $isOnBranch$, we solved the problem, because when $isOnBranch$ is true, only the first iterated array in the next level is on the branch. Once the search turns back and goes up, all other arrays are not on the search branch but at the left of the branch.

The $currentLevel$ can indicate which level the operating array is at. It is used to get the right remainder from the $remaindersArray$. If the array is on the branch, the $remainder$ could determine which pointers it needs to start searching.

\subsection{Deletion}
When deleting a value, after the value is found, it will be deleted. After it is deleted, the process needs to check the whole array to see if it can be freed. If the data array does not contain other values and could be freed, the process will find the left data array and free it from the data array list. If you adopt a bidirectional link implementation at the data level, the left data array can be accessed directly; otherwise, you need to use the mechanism in the insertion to find it. If we freed the array, we need to check the index array in the above level to see if I can continuously free that array.

Compared with B+-tree and binary search trees, LAT does not require moving elements and has fewer operations. In the B+-tree, the insertion and deletion may invoke the complex steps of splitting and merging, and the binary search tree needs to rotate to keep the balance. The B+-tree needs to load an entire block in each level, while the LAT only needs to read one node in each level. LAT not only offers a direct access approach to locate the value, but also provides a simpler implementation.

\subsection{Optimization}
We already demonstrated the initial LAT, which can locate the element directly. However, there is still a big problem. In modern environments, the length of a key could be 64 bits or even more. Involving such a large number in division to calculate the remainders is not efficient. Hence, to achieve a higher performance, we need to optimize the remainder calculation process. Fortunately, there is a way to totally avoid it, and the core is picking a suitable $radix$ and $height$.

\begin{lemma}
Let $n$ be the $key$'s length, $radix$ $r$ be a power of 2, $r=2^p$, and $h$ be the $height$. If $n=p \times h$, then the data from $i$ bit to $i+p-1$ bit in the key is the $remainder\ \frac{i}{p}$, $(\frac{i}{p} \in \mathbb{Z}_{\geq 0})$.
\end{lemma}
\begin{proof}
Suppose the key $k$ is composited by bits $k_{0}, k_{1}, k_{2}, \cdots, k_{n-1}$ from left to right, high bits at left. $k=\sum_{i=0}^{n-1} k_{i} \times 2^{n-1-i}$. When $i \leq n-p-1$, $n-1-i\geq p$, so $k_{i} \times 2^{n-1-i}$ is a multiple of $2^p$. Then we have:

\begin{align}
k&=k_{0} \times 2^{n-1} + k_{1} \times 2^{n-2} + \cdots + k_{n-p-1} \times 2^{p} + k_{n-p} \times 2^{p-1} \cdots + k_{n-1} \times 2^{0} \\
& \equiv 0+0+\cdots+0+ k_{n-p} \times 2^{p-1}+\cdots + k_{n-1} \times 2^{0} \pmod{2^{p}} \\
& = k_{i} \times 2^{p-1}+ k_{i+1} \times 2^{p-2}+ \cdots + k_{i+p-1} \times 2^{0} \quad (i=n-p)
\end{align}

Hence, we have provided the rightmost $p$ bits data in $k$ is equal to the remainder that is calculated out in Algorithm 1, line 2. Then we will prove $\lfloor k {\div} r \rfloor$ is equal to $k$ right shifts $p$ bits.

According to the definition of modulo operation, $\lfloor k {\div} r \rfloor = (k- (k \mod r))\div r$. Since we already proved the last $p$ bits in $k$ equal to $k\mod r$, we have:

\begin{align}
\lfloor k {\div} r \rfloor &= (k- (k \mod r))\div r \\
&=(k_{0} \times 2^{n-1} + k_{1} \times 2^{n-2} + \cdots + k_{i-1} \times 2^{n-i} + k_{i} \times 2^{n-1-i}) \div 2^p \quad (i=n-p-1) \\
&=(k_{0} \times 2^{n-1} + k_{1} \times 2^{n-2} + \cdots +k_{i-1} \times 2^{p+1} + k_{i} \times 2^{p}) \div 2^p \\
&=k_{0} \times 2^{i} + k_{1} \times 2^{i-1} + \cdots +k_{i-1} \times 2^{1} + k_{i} \times 2^{0}
\end{align}

Hence, $\lfloor k {\div} r \rfloor$ is equal to $k$ shifts $p$ bits right. As $\lfloor k {\div} r \rfloor$ will become the new $k$, the next remainder calculation will get the last $p$ bits from the new $k$. To calculate all remainders out, it needs $\frac{n}{p}$ iterations and $h=\frac{n}{p}$. Because the remainders that are calculated out earlier will be used later, let a remainder be calculated out in $i$'th round, then it will be used at $level\ h-i$, and it is the $remainder\ h-i$. Since The remainder calculation is from behind to the head, the $i$ bit to $i+p$ bit in $k$ will be calculated out in $(h- \frac{i}{p})$'th round $(\frac{i}{p} \in \mathbb{Z}_{\geq 0})$, and it is the $remainder\ \frac{i}{p}$.

Therefore, the lemma has been proved.
\end{proof}

By adopting this concept, we can avoid calculating the remainders in a loop but simply change the viewpoint. Figure 4 provides an example. In this example, the key's $length$ is 8, the $radix$ is 4, and $height$ is 4. So, the initialization process of LAT could be modified to: pick a power of 2 as the $radix$, $2^p$, get the bit length $n$ of the $key$, and set the $height$ to $\lceil \frac{n}{p}\rceil$.
\begin{figure}[htbp]
\centerline{\includegraphics[width=\linewidth]{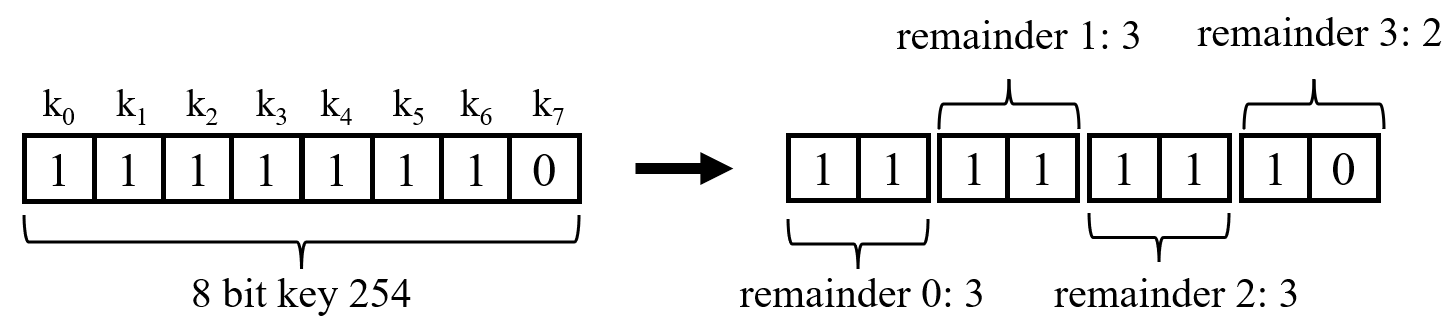}}
\caption{Optimized remainders calculation.}
\end{figure}

\subsection{Key Properties of LAT}

The Linked Array Tree (LAT) presents several structural and computational advantages over traditional search data structures. These properties support its applicability in big data and concurrent environments:

\begin{itemize}
    \item \textbf{Constant-Time Access:} Lookup, insertion, and deletion operations have $O(1)$ complexity, assuming fixed radix and height, regardless of the number of stored elements.
    
    \item \textbf{Non-Moving Updates:} Insertion and deletion do not require moving existing elements, avoiding rotations or rebalancing typically seen in tree structures.
    
    \item \textbf{Low Memory Overhead:} LAT avoids explicit key storage and reduces pointer usage, leading to efficient memory consumption, especially in dense or sequential data scenarios.
    
    \item \textbf{Parallelization-Friendly:} Due to its direct access paths and minimal structural modification, LAT is naturally suited to concurrent and parallel operations.
    
    \item \textbf{Bitwise Optimization:} When radix is a power of two, remainder calculations can be replaced with bit operations, further improving access speed.
    
    \item \textbf{Scalable and Flexible:} The structure supports dynamic growth by increasing the height when needed, enabling it to scale efficiently with key space expansion.
\end{itemize}

\section{Time Complexity Analysis}
\subsection{Search}
When searching for a value, the travel times between nodes are constant. Let $I_{s}$ be the iteration times in search,$h$ be$height$, then $I_{s}=h$. Let $m$ be the $max\ size$, according to Lemma 2.1, $I_{s}=h=log_{r}(m)$.
\begin{theorem}
    Let $T_{s}$ be the search time complexity of a LAT, then$$T_{s}=O(log_{r}(m))$$
\end{theorem}
\begin{proof}
    Since the iteration times in search are equal to $h$, and $h$ is constant in a LAT, the iteration times in search is a constant number. Hence, the search time complexity in LAT is constant and is not impacted by the number of stored elements.
\end{proof}
\subsection{Insertion}
When inserting a value, there are 2 cases. In the complex case, the data array is newly allocated, and we need to find the left data array. Let $I_{f}$ be the iteration times in finding the left data array, then the insertion iteration times $I_{i}$, in this case, is $I_{i}=I_{s}+I_{f}$. Figure 5 shows an example of the worst case and best case in finding the left data array.
\begin{figure}[htbp]
\centerline{\includegraphics[width=\linewidth]{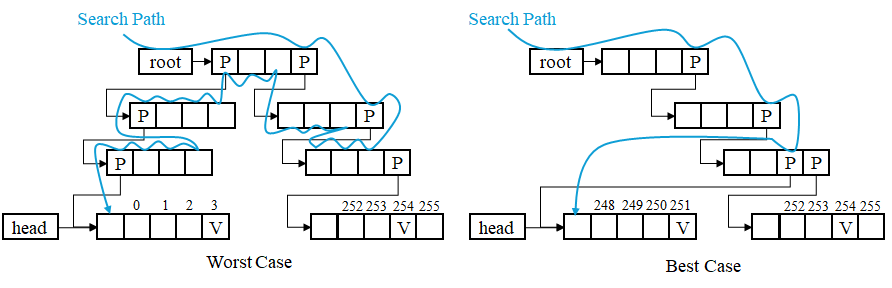}}
\caption{Worst case and best case in finding left data array.}
\end{figure}

\begin{theorem}
    Let $T_{i}$ be the insertion time complexity of a LAT, then$$T_{i}=O(log_{r}(m))$$
\end{theorem}
\begin{proof}
    Let $I_{b}$ be the iteration times in finding the left data array in the best case. Let $I_{w}$ be the iteration times in finding the left data array in the worst case.
    $$I_{w}=(2height -3)\times radix$$

    In the worst case, the search path need to search near all nodes in $2(height-2)+1$ index arrays. We notice it would not access the newly allocated array, but it needs to travel to the left data array in the last step. So, the iteration time is equal to the number of all nodes in the index arrays on the path. When the insertion needs to find the left data array, we have:
    \begin{gather}
        I_{b}\leq I_{f}\leq I_{w} \\
        I_{s}+I_{b} \leq I_{i}\leq I_{s}+I_{w} \\
        2\times height \leq I_{i}\leq height + (2height -3)\times radix
    \end{gather}
    
    Since the $height$ and $radix$ are constant in a LAT tree if $max\ size$ does not change, in all cases, the insertion iteration times $I_{i}$ is not bigger than a constant number.
\end{proof}

\subsection{Deletion}
The early part of the deletion process is like in the search, but it needs to check the arrays on the search branch to see if the arrays could be freed. So, let $I_{d}$ be the iteration times in deletion, $I_{c}$ be the iteration times in the clear process. Then, if the deletion needs to free the array, it needs to find the left data array, $I_{d}= I_{s}+I_{f}+I_{c}$.
\begin{theorem}
    Let $T_{d}$ be the deletion time complexity of a LAT, then$$T_{d}=O(log_{r}(m))$$
\end{theorem}
\begin{proof}
    Let us consider the worst case. Let $r$ be the $radix$, $h$ be the $height$. In the worst case, $I_{f}=I_{w}=r\times(2h -3)$. Because there is no other value or pointer in the arrays on the branch in the levels that are bigger than $level\ 0$, the check process needs to check all nodes, then in the worst case in Figure 5  $$I_{c}=h \times r -r +1$$

    Mind that in $level\ 0$, there is another pointer. If there is no other pointer, the finding left process would be much shorter. In all cases, the iteration times in $level\ 0$ in finding the left process plus in the clear process is not bigger than $r+1$. In the worst case, $$I_{d}= I_{s}+I_{w}+I_{c}= h+ (2h -3)\times r +h \times r -r +1 =3hr-2r+h+1$$

    Then, in all cases:$$I_{d}\leq 3hr-2r+h+1$$
    Since $r$ and $h$ are constant in a LAT, all deletion iteration times are not bigger than a constant number.
\end{proof}

\section{Performance Comparison}
In this section, we will present a performance comparison between the new structure, the red-black tree, and the B+ tree, since the red-black tree is a representative binary search tree, and the B+ tree is pervasively used in large systems. The selected implementation of the red-black tree is the C++ STL map because it is widely used in the industry and can reflect the performance of the modern red-black tree. The implementation of B+-tree is the C++ TLX btree\underline{ }map \cite{13}. TLX offers a collection of data structures and algorithms that are not included in STL, and it can represent the performance of a modern implementation of B+-tree.

In our real world, most usage scenarios are a combination from:
\begin{itemize}
\item data intensive insertion: Insert a large chunk of continuous data into the structure, e.g. writing the value of keys $0, 1, 2, \dots, n-1$ into the structure.
\item data intensive search: Search some continuous keys, e.g. searching keys $0, 1, 2, \dots, n-1$.
\item data intensive deletion: Delete some continuous data, e.g. deleting keys $0, 1, 2, \dots, n-1$.
\item Data sparse insertion: Insert the value of some random keys into the structure; they are sparse and irrelevant.
\item Data sparse search: Search the value of some random keys in the structure; they are sparse and irrelevant.
\item Data sparse deletion: Delete some data randomly one by one.
\item data intensive iteration: When the data set in the structure is a large continuous chunk, iterate all existing data. This operation can indicate the sequential operation performance, like range search.
\item Data sparse iteration: When the data set in the structure is sparse and irrelevant, iterate all existing data.
\end{itemize}

The speed performance test uses a 64-bit key and and 64-bit value, programmed with C++ language, running on an AMD Ryzen 7 7840Hw CPU.

\begin{figure}[htbp]
\centerline{\includegraphics[width=\linewidth]{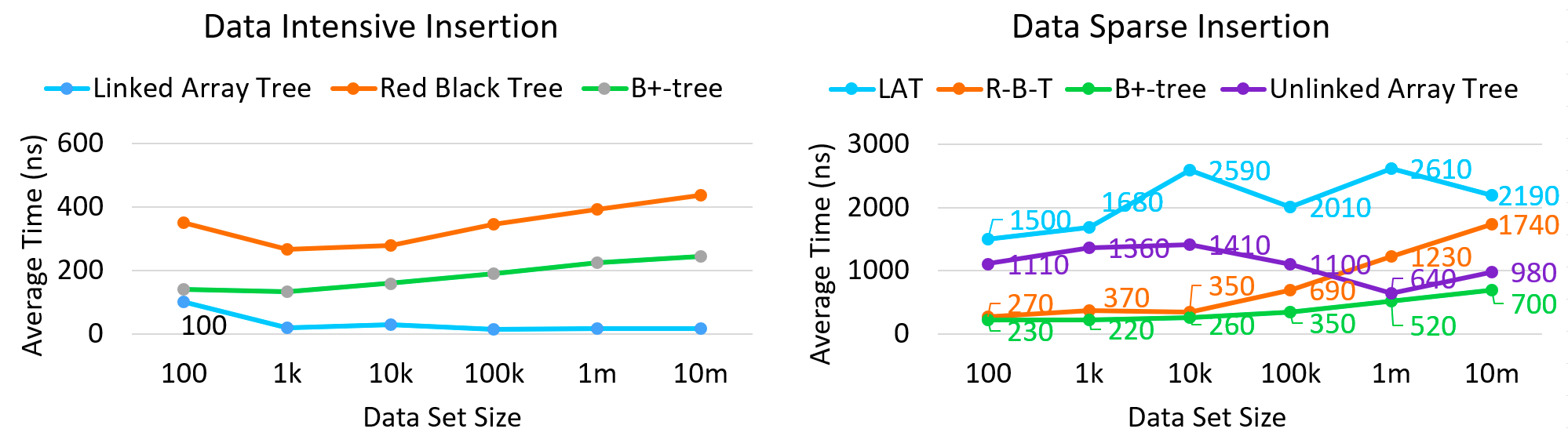}}
\caption{Insertion test for LAT, red-black tree, B+-tree.}

\end{figure}

\begin{figure}[htbp]
\centerline{\includegraphics[width=\linewidth]{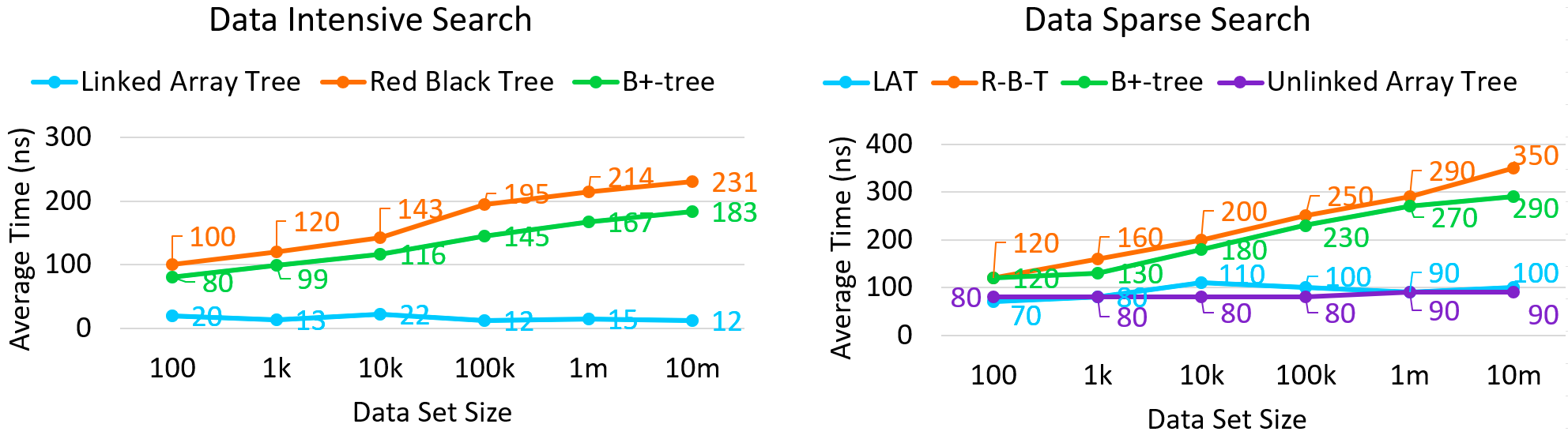}}
\caption{Search test for LAT, red-black tree, B+-tree, unlinked array tree.}
\end{figure}

\begin{figure}[htbp]
\centerline{\includegraphics[width=\linewidth]{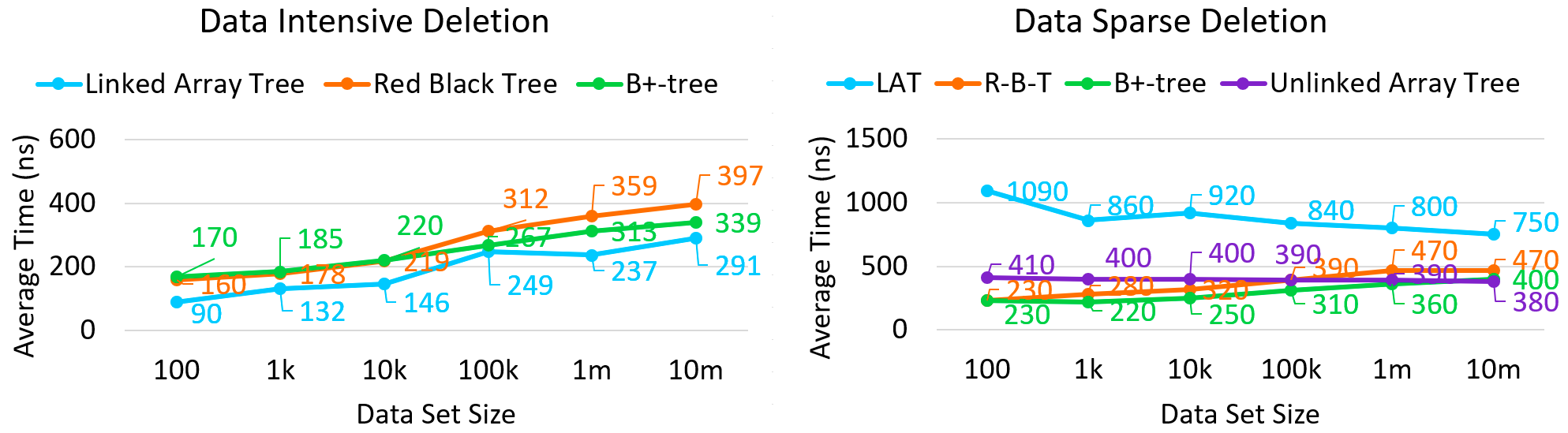}}
\caption{Deletion test for LAT, red-black tree, B+-tree, unlinked array tree.}
\end{figure}

\begin{figure}[htbp]
\centerline{\includegraphics[width=\linewidth]{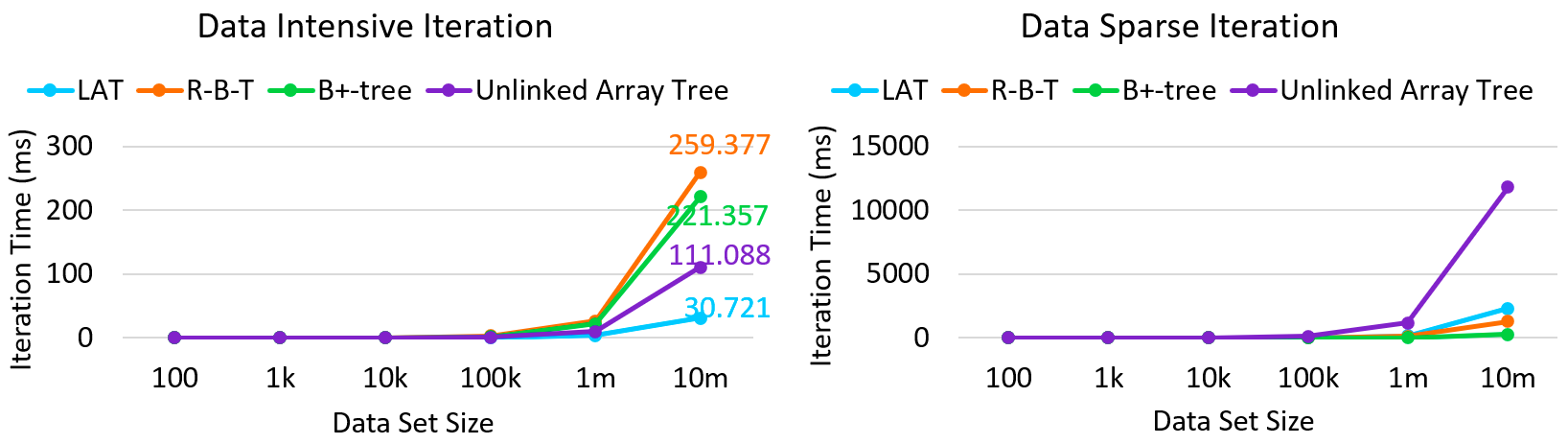}}
\caption{Iteration test for LAT, red-black tree, B+-tree, unlinked array tree.}
\end{figure}

In the data intensive test, the LAT provides an outstanding performance as it can insert or retrieval the data within 15ns in 100 million data. However, in the data sparse test, in the insertion and deletion, LAT takes a longer time compared with traditional algorithms. If we do not implement the pointers in data level, the insertion and deletion process do not need to find left array. The unlinked array tree halves the insertion and deletion time in the data sparse test, but this sacrifices the sequential operation performance as the trade-off.

The key strength of LAT is that it can guarantee a constant access time both in data intensive and data sparse test. Irrelevant to the stored data size, the performance of it can be expected. So, it is suited for the data intensive work in big data. For small, sparse data sets, the traditional search algorithms still have advantages.

\begin{figure}[htbp]
\centerline{\includegraphics[width=\linewidth]{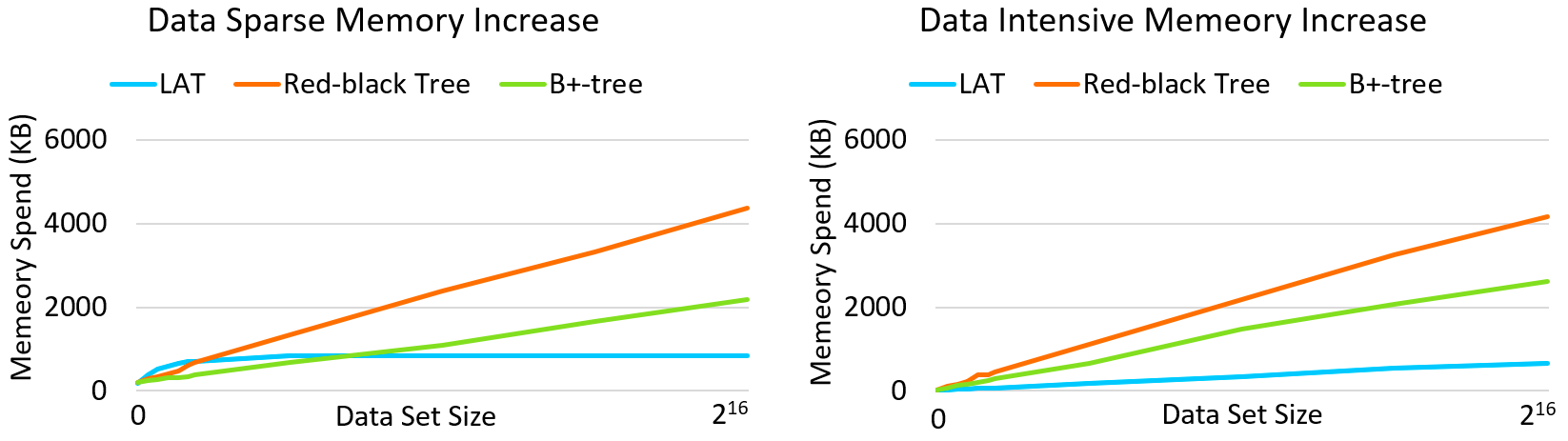}}
\caption{Memory increase test for LAT, red-black tree, B+-tree.}
\end{figure}

In memory increase test, the key length is 16 bits so that we can cover all keys and see the overall curve. In the data intensive test, memory consumption increasing curses of all algorithms are linear, but LAT is much less than the red-black tree and the B+-tree because it has fewer pointers and does not need to store keys. In the sparse test, the curves of the red-black tree and B+-tree are linear, while the LAT is steep at the beginning but flat later. If the data set takes a high proportion of max size, LAT requires much less memory than other algorithms; but if the data set is relatively small, it takes much more memory than others. This phenomenon is prominent when the key length is long. LAT could take tens of times more memory than other algorithms at the beginning. So if the data set is small and randomly distributed, LAT should be avoided.

For small random data sets, we recommended using a small $radix$ since the big $radix$ will waste many spaces in this scenario, while for a large chunk of continuous data, the only recommended $radix$ is 256 because modern memory stores data as bytes. To mitigate the memory waste problem in LAT in sparse usage, continuous optimization is required. For example, a LAT could be asymmetric, meaning the $radix$ in each level could be different. In Figure 9, the asymmetric unlinked array tree has 8 nodes in level 0, and it uses the highest 3 bits in the key to locate the the pointer in this level. It takes 30 nodes to store the data. If it is symmetric with $4$ as the radix, it will need 32 nodes.

\begin{figure}[htbp]
\centerline{\includegraphics[width=0.6\linewidth]{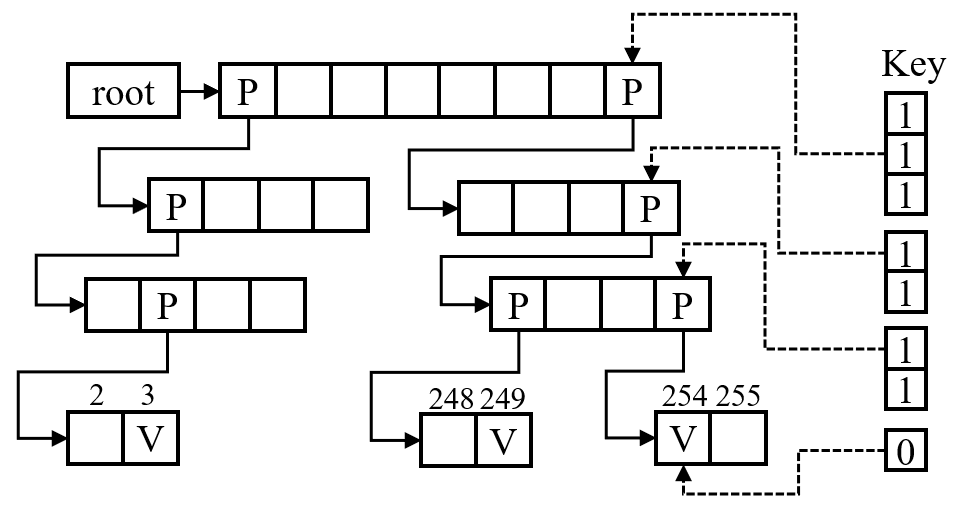}}
\caption{The asymmetric unlinked array tree has different radix in each level.}
\end{figure}

\section{Conclusion}
The search concept in LAT is like the paging mechanism in memory management, but this paper has a deeper analysis of the relationships between the bits in the key and the search path. Each key can be mapped to a unique path, and the structure is not necessary to be symmetric.

The linked array tree, which can achieve a constant time complexity, is not a replacement of the traditional search algorithm. Overall, LAT is suitable for a large, continuous chunk of data, like the memory management or disk management, or a large random data set that can take a high proportion of the max size. In these scenarios, LAT provides less storage access times, lower memory overhead, and the elimination of node movement. However, the small and sparse data set should avoid it because of the high memory overhead in the beginning.
Compared with the array, it does not allocate all space at once. It may be suitable for sparse matrix calculation. Through the $radix$ of LAT, it can offer a flexible way to control the trade-off between time and space.

The relevant code and test data now is available at: \href{https://github.com/SongpengLiu/LinkedArrayTree}{https://github.com/SongpengLiu/LinkedArrayTree}

\end{document}